%% file: main.tex
\documentclass[]{article}

\input{macros}

\input{macros.local}

\title{Random mappings designed for commercial search engines}
\author{%
Roger Donaldson\thanks{Department of Mathematics, The University of British
Columbia}, %
Arijit Gupta\thanks{Department of Mathematics, Imperial College London}, %
Yaniv Plan\footnotemark[1] \thanks{YP is partially supported by NSERC grant 22R23068}, and %
Thomas Reimer\footnotemark[1]%
}

\begin{document}
\maketitle

%% frontmatter.tex

\begin{abstract}
We give a practical random mapping that takes any set of documents
represented as vectors in Euclidean space and then maps them to a sparse subset
of the Hamming cube while retaining ordering of inter-vector inner products.  Once
represented in the sparse space, it is natural to index documents using
commercial text-based search engines which are specialized to take advantage of
this sparse and discrete structure for large-scale document retrieval.  We give
a theoretical analysis of the mapping scheme, characterizing exact asymptotic
behavior and also giving non-asymptotic bounds which we verify through
numerical simulations.  We balance the theoretical treatment with several
practical considerations; these allow substantial speed up of the method.  We
further illustrate the use of this method on search over two real data sets: a
corpus of images represented by their color histograms, and a corpus of daily
stock market index values.
\end{abstract}

\section{Introduction}

Rapid document retrieval is a basic problem in modern large-scale computation.
One is given a corpus of documents and a query, and the goal is to find
documents in the corpus that nearly match the query.  We consider the case where the query is itself a document and thus our problem is closely related to \textit{approximate nearest neighbor search}.  While this problem becomes challenging in the high-dimensional setting, in the special case of text documents it has been addressed, in practice, with success.  Text documents may be represented as ``bags of
words'' in which each document is represented as a vector, whose length is the size $M$
of the dictionary, and whose $i$-th entry contains the number of times word
$i$ appears.\footnote{For example, the first sentence of the introduction would
be represented as a vector with 12 entries equal to one and all other entries
equal to zero.}  Commercial search engines have been built to take advantage of
this sparse and discrete structure for large scale search.

However, many documents of interest are best represented as vectors in
$\real^d$, which are not amenable to large scale search in their original form.
A small set of examples includes images, video, and audio; biometric data such
as fingerprints and iris scans; and bioinformatic data such as gene expressions,
enzyme activities, and protein libraries.  Given the high degree to which
text-based search has been implemented, including the consequential familiarity
of software engineers for text-based search infrastructure, we give a way to
adapt this familiar infrastructure to search in $\real^d$.

Naive search for objects in $\real^d$ for $d \gtrsim 50$ is prohibitively
slow when the corpus is large.  Indeed, as stated in~\cite{andoni2008near-optimal}, ``Despite decades of
intensive effort, the current solutions suffer from either space, or query time,
that is exponential in dimension $d$''.  See~\cite{samet2005foundations} for an
overview. A typical solution is to apply a dimension-reduction technique, such
as that espoused by Johnson Lindenstrauss Lemma~\cite{JL_lemma}, which maps the
data into a lower dimensional Euclidean space.  While this has a beautiful
theoretical backing, in practice Euclidean space is not ideally suited for
search.  Many other mapping strategies have been invented, in particular
the \textit{locality-sensitive hashing} used in approximate nearest neighbour
search \cite{indyk1998approximate, gionis1999similarity}.  These give
theoretically appealing low-complexity and low-storage guarantees, but are
typically dissimilar in operation to text-based search, so cannot leverage
existing search infrastructure, nor be naturally combined with text-based
search.  In this paper, we give an easy to use practical mapping which puts
the data precisely in the form used by commercial text search engines, namely in
high-dimensional but \textit{sparse} subsets of $\integer^m$ or even of $\{0,
1\}^m$.  Mapping real valued vectors to discrete measurements is also the approach of locality sensitive hashing.  However, in this paper we emphasize mapping to a sparse set to put the data into the same representation as a text document.  We believe this falls outside of the perspective considered in prior literature.

\subsection{Random mapping scheme}

Although proving that our random mapping scheme works is involved, the scheme
is remarkably simple.  Our corpus $\X$ is a finite collection of vectors in
$\real^d$, normalized to have unit $\ell_2$ norm.  To transform each vector in
$\X$, multiply each vector by a random matrix, then threshold each element.  We
now formalize this procedure.

Introducing notation we will use throughout, let $a_1, a_2, \hdots, a_m$ be
standard normal random vectors of length $d$.  Fix $h > 0$.  Map $x \in\X
\subset \sphere^{d-1}$ to the Hamming cube $\{0,1\}^m$ as follows 
\begin{equation}
    x \rightarrow (\one_{[\ip{a_i}{x} \geq h]})_{i =1}^m.
\end{equation}
Above, $\one_{[\ip{a_i}{x} \geq h]}$ is equal to $1$ if $\ip{a_i}{x} \geq h$ and $0$ otherwise.  Note that $x$ is mapped to a sparse vector provided $h$ is
large enough.  

After indexing each document in this manner, we search by performing the same
transformation to query vector $y \in \sphere^{d-1}$.  We then take inner
products in the Hamming cube to determine the best match.  Thus, we return $x
\in \X$ to the user in decreasing order of score, \begin{equation}
    \score{x}{y} = \frac{1}{m}\sum_{i=1}^m \one_{[\ip{a_i}{x} \geq h]} \one_{[\ip{a_i}{y} \geq h]}.
    \label{eq:score}
\end{equation}
The bulk of this paper is devoted to proving and demonstrating that with
appropriate choices for parameters $h$ and $m$, this procedure returns $x \in
\X$ closest, or nearly closest, to query $y$ with high probability.

The astute reader will already observe that the transformation itself is
$\bigO(m d)$:  Although the theory is straightforward to introduce with $a_i$
Gaussian, we elaborate in the section on practical considerations,
Section~\ref{s:practical}, on an alternate choice for the random transformation
that is $\bigO(m \log m)$.  We show through simulations that this alternate strategy gives identical performance to that observed with a Gaussian transform.

The particular case where $h=0$, correponding with measurements of the form sign($\ip{a_i}{x}$), was introduced as a method for locality-sensitive hashing in \cite{Charikar2002}, and is well studied in general.  Similarly,
other authors~\cite{Jiang2015, KulisB.aGrauman2009} threshold to $h=0$, but
replace $\ip{a_i}{x}$ with an inner product with respect to a kernel tailored to
the corpus $\X$ in order to gain greater fidelity in the hashed (transformed)
space.  Assuming $\X \subset \sphere^{d-1}$, the vanilla sign mapping has been shown to be a near isometry from the sphere with geodesic distance to the Hamming cube, even in the case when $\X$ has an infinite collection of elements \cite{plan2014dimension}.  Thus, it is a very effective locality-sensitive hash.  Further, this near-isometric property is pivotal in \textit{1-bit compressed sensing} \cite{boufounos20081, plan2013robust}.  In contrast, we focus on the case when $h$ is much greater than 0.  We show that, while the mapping is not a near isometry, it still does preserve ordering of inner products, i.e., if $\ip{x}{y} > \ip{z}{y}$ then, with high probability, $\score{x}{y} > \score{z}{y}$.  This relationship is sufficient for effective document search and by relaxing the need for near-isometry we gain sparsity and thus improve search speed.
 %We gain this greater fidelity by considering $h>0$, and in doing so,
%incorporate a relaxed requirement that the probability of hash collisions
%measured by~\eqref{eq:score} be merely monotonically increasing with respect to
%$\ip{x}{y}$.  
As it happens, our choice of $h>0$ also provides increased
sensitivity to distinguishing the nearly-similar items we in practice expect to
find among the top search results.

\subsection{Notation}

$\sphere^{d-1}$ refers to the Euclidean sphere in $d$ dimensions; the corpus
$\X \subset \sphere^{d-1}$ is a set of documents to be indexed (we assume each
document has been normalized);  $x \in \X$, $y \in \sphere^{d-1}$ always refer
to a corpus document, and a user query, respectively; $\mu(\lambda) = \E S(x,y)$
is the expected score of query $y$ against document $x$ given their true
similarity as their inner product $\lambda = \ip{x}{y}$; the vocabulary size
$m$ is a large parameter that we control.

$\Phi(t) = P(N(0,1) < t)$ is the cumulative distribution function of a standard
normal random variable;  $a_1, a_2, \hdots, a_m$ are standard normal vectors,
each of length $d$; the threshold value $h$ is parameterized by a scalar $r$ so
that $h = h(m,r) := \sqrt{2r \log m}$.

For two sequences of numbers, $b_1, b_2, \hdots$ and $\beta_1, \beta_2, \hdots$,
we say that $b_m$ is \textit{asymptotically equivalent} to $\beta_m$, denoted
$b_m \sim \beta_m$, if and only if \[b_m = \beta_m(1 + o(1))\]
where $o(1) \rightarrow 0$ as $m \rightarrow \infty$.  If $\beta_m \neq 0$ for each $m$, then $b_m \sim \beta_m$ if and only if
\[\lim_{m \rightarrow \infty} \frac{b_m}{\beta_m} \rightarrow 1.\] Note that
this (standard) equivalence relationship is preserved under addition,
multiplication, and division. 

%% definition.tex

\section{Problem definition}
\label{sec: definition}

As previously noted, there exists a great body of literature addressing nearest
neighbor and approximate nearest neighbor search.  We broaden this perspective
slightly to address the more practical problem of producing a list of results,
ordered by decreasing relevance.  We therefore introduce notions of relevance,
retrieval, and errors in order to use standard information retrieval performance
metrics (see \cite{manning2008introduction} for an overview).

Recall that $\mu(\lambda) := \E S(x,y)$ provided $\ip{x}{y} = \lambda$.  Thus,
with a goal of retrieving documents whose inner product with $y$ is at least
$\lambda$, we will return any document whose score exceeds $\mu(\lambda)$.

\begin{definition}[Document sets] Fix $\lambda \in [-1,1]$, $y \in \sphere^{d - 1}$, and $x \in \X$.  We call the document $x$ \emph{retrieved} if and only if $\score{x}{y}\geq\mu(\lambda)$.  We call the document $x$ \emph{relevant}  if and only if
        $\ip{x}{y}>\lambda$; a document that is not relevant is \emph{irrelevant}.
        \end{definition}

While purely relevant documents are appealing, in practice the mapping will
cause some (small) error.  Thus, we replace the notion of relevant with
\textit{$\epsilon$-relevant}.

\begin{definition}[Document sets with small error] Fix $\lambda \in [-1,1],
    \epsilon > 0$, $y \in \sphere^{d - 1}$, and $x \in \X$.   We call the document $x$ \emph{$\epsilon$-relevant} if and only if $\ip{x}{y} \geq \lambda + \epsilon$; we call the document $x$ \emph{$\epsilon$-irrelevant} if and only if $\ip{x}{y} \leq \lambda - \epsilon$.
    \end{definition}

Borrowing notions from hypothesis testing, we identify two important events:
\begin{definition}[Error events]  We define two types of events:  If an $\epsilon$-irrelevant document is retrieved this is called a \emph{type I error}.  If an $\epsilon$-relevant document is not retrieved this is called a \emph{type II error}.
\end{definition}

\noindent See Figure \ref{figure: error_diagram} for an illustration.
\begin{figure}
\centering
\includegraphics[width = .7 \textwidth]{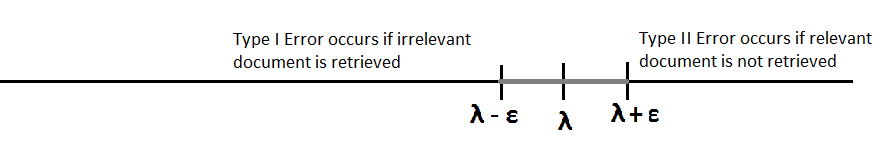}
\caption{Diagram of kinds of errors.  In the area between $\lambda - \epsilon$
    and $\lambda + \epsilon$, we do not expect documents to be returned
    reliably.  We give a precise characterization of the size of $\epsilon$
    required to keep the probability of type I errors and type II errors small
    in Section \ref{sec: main theory}}.  \label{figure: error_diagram}
\end{figure}

The practitioner will have two main goals:
\begin{enumerate}
    \item Minimize the complexity of retrieval
        in space (precomputation, also called indexing) and time (search) by making
        the mapping as sparse as possible; and
    \item Minimize $\epsilon$ while still
        controlling the probability of type I errors and type II errors.
\end{enumerate}
Not surprisingly, there is a tradeoff between those two goals. Our main
theoretical contribution, given in the next section, is to precisely
characterize the sparsity and the size of $\epsilon$ as a function of adjustable
parameters and given the desired bound on type 1 errors and type II errors.

%% theory.tex

\section{Theoretical results}
\label{sec: theory}

This section states results concerning sparsity of our mapping and
asymptotic and non-asymptotic statements of its accuracy.  These results
anticipate the roles of parameters $h$ and $m$ in the tradeoff between accuracy
and sparsity (hence, complexity) elucidated in Section~\ref{s:practical}.

\subsection{Sparsity of the map}
\label{sec: sparsity}

We begin by controlling the expected sparsity of mapped vectors.  Recall that
each vector $x \in \X$ is mapped to the sequence $\one_{[\ip{a_i}{x} > h]}, i =
1, 2, \hdots, m$.  Observe that $\ip{a_i}{x}$ is standard normal, and thus the
expected number of non-zero components of a mapped vector is $m P(N(0,1) \geq h)
= m (1 - \Phi(h)).$  Recall that $h = \sqrt{2 r \log m}$.  For large $m$, which
implies large $h$, one may use the classic \cite{gordon1941values} approximation
\begin{equation}
    \label{eq: Gaussian tail}
P(N(0,1) \geq h) \sim \frac{e^{-\frac{h^2}{2}}}{h \sqrt{2 \pi}} = \frac{m^{-r}}{\sqrt{4 \pi r \log m}}.
\end{equation}  
Note that the right-hand side is not only asymptotically accurate, but is also a non-asymptotic upper bound on $P(N(0,1) \geq h)$ \cite{gordon1941values}.  

To conclude, let $k$ be the number of non-zero entries in a mapped vector.  Then $k$ satisfies
\begin{equation}\label{eq: expected non-zeros in map}
\E \, k = m (1 - \Phi(\sqrt{2 r \log m})), \qquad \E \,k \sim \frac{m^{1-r}}{\sqrt{4 \pi r \log m}}, \qquad \text{and} \qquad \E \,k \leq \frac{m^{1-r}}{\sqrt{4 \pi r \log m}}.
\end{equation}
Thus, one sees that after ignoring logarithmic factors there are roughly $m^{1-r}$ non-zero entries.  In fact, it is for this simple expression, and similar expressions below, that we parameterize $h$ as $h = \sqrt{2 r \log m}$.  Indeed, in Lemma \ref{lem: Gaussian score} below, we see that $r$ naturally characterizes a phase transition.

We now proceed to the more significant part of the theory: characterizing the size of the error.  This will come from controlling the concentration of the score around its mean.
We begin in the asymptotic regime to give an exact result and a useful intuition.

\subsection{Asymptotic characterization of the score}
We begin with the following lemma, concerning the asymptotic normality of $S(x,y)$.

\begin{lemma} \label{lem: Gaussian score}
Fix $x, y \in \sphere^{d-1}$ with $\ip{x}{y} = \lambda < 1$.  Let $\mu = \mu(\lambda)$ and $\sigma^2 = \sigma^2(\lambda)$ be, respectively, the expectation and variance of $\one_{[\ip{a_i}{x}]} \one_{[\ip{a_i}{y}]}$.  Depending on $\lambda$, there are two cases to consider.

\noindent{\bf Case 1:}  $\lambda < 2r - 1$.  

Then 
\[P(S(x,y) \neq 0) \leq \mu(\lambda) m \rightarrow 0 \qquad \text{as} \quad m \rightarrow \infty.\]

\noindent{\bf Case 2:}  $\lambda > 2r - 1$.

The normalized score
\begin{equation}\label{eq: tilde S}
\tilde{S}(x,y) := \frac{\sqrt{m}(S(x,y) - \mu)}{\sigma},
\end{equation}
convergences to a standard normal random variable in distribution as
$m\rightarrow \infty$.

\end{lemma}  

The phase transition at $\lambda = 2r - 1$ can be intuited by the following
observations.  When $\lambda < 2r -1$, the expected number of non-zero summands
in $S(x,y)$ converges to 0.  When $\lambda > 2r -1$, the expected number of
non-zero summands converges to infinity, and the score exhibits Gaussian
behaviour.  Let us derive this precisely.  To ease notation, here and below let
$\ip{a_i}{x} = w_i$ and $\ip{a_i}{y} = v_i$.  Let $w, v$ be standard normal with
covariance $\lambda$ so that $w_i, v_i$ are independent copies of $w, v$.
Consider $\E m S(x,y) = m \mu$, that is, the expected number of non-zero
summands in the score.  To control this quantity, we need the following
bivariate normal tail approximation, which can be derived from
\cite{savageGaussianBounds}, \begin{equation}
\mu = P(\one_{[w > h]} \one_{[v > h]} = 1) \sim \frac{(1 + \lambda)^2}{2 \pi h^2 \sqrt{1 - \lambda^2}} \exp\left(- \frac{h^2}{1 + \lambda}\right) = C(\lambda) \frac{m^{- \frac{2r}{1 + \lambda}}}{2 r \log m}, \qquad C(\lambda) := \frac{(1 + \lambda)^2}{2 \pi \sqrt{1 - \lambda^2}}.
\end{equation}
Thus, the expected number of non-zero summands is 
\begin{equation}\label{eq: expected nonzeros}
\E m S(x,y) = m \mu \sim C(\lambda) \frac{m^{ \frac{\lambda - (2r - 1)}{1 + \lambda}}}{2 r \log m}.
\end{equation}
It is now clear that this quantity converges to infinity for $\lambda > 2r -1$
and converges to 0 for $\lambda < 2r - 1$.  This latter observation, combined
with Markov's inequality, already completes the proof of the lemma in Case 1.
Indeed, Markov's inequality shows that \begin{equation}\label{eq: use Markov}
P(S(x,y) \neq 0) = P(m S(x,y) \geq 1) \leq \E m S(x,y) = m \mu(\lambda) \rightarrow 0.
\end{equation}
The Gaussian behaviour of Case 2 does not follow from the vanilla central limit
theorem since the summands depend on $m$ (through $h$), but nevertheless is
proven as an application of a Berry-Esseen approximation.  We state this
approximation, and complete the proof, in Appendix~\ref{ssec: gaussian score
proof} below. 

\subsection{Main theoretical results}
\label{sec: main theory}
We now leverage Lemma \ref{lem: Gaussian score} to determine the efficacy of
 search via the sparsifying transformation of this paper.  In the following
theorem we give a precise asymptotic characterization of $\epsilon$, which
vanishes as $m$ increases.  Below, $\eta$ is a parameter which controls the
expected number of errors.

\begin{theorem}[Main theorem, asymptotic version]
\label{thm: main}
Fix $y \in \sphere^{d-1}, \lambda \in (2r -1, 1)$ also satisfying $\lambda > 0$ and $\eta > 0$. Let 
\begin{equation}
\label{eq: epsilon}
\epsilon = \epsilon(\lambda, m, \eta, r) = C(\lambda, r, m, \eta)\, m^{- \frac{\lambda - (2r - 1)}{2(1 + \lambda)}}, \qquad  C(\lambda, r, m, \eta) := \frac{\sqrt{2 \pi}\,(1 + \lambda)(1 - \lambda^2)^{1/4} }
{\sqrt{2 r \log m}} \eta.
\end{equation}

Then,
\[\lim_{m \rightarrow \infty} \sup_{|\X| = n}  \frac{\E (\text{Number of type I errors}) + \E (\text{Number of type II errors})}{n} = P(N(0,1) \geq \eta).\]
Thus, by taking $\eta = 2 \sqrt{\log n}$, we have
\[\lim_{m \rightarrow \infty} \sup_{|\X| = n}\frac{ \E (\text{Number of type I errors}) + \E (\text{Number of type II errors})}{n} \leq \frac{1}{n^2 \sqrt{8 \pi \log n}}.\]
\end{theorem}
The last inequality follows from the Gaussian tail bound \eqref{eq: Gaussian tail}.  Note that the last inequality combined with Markov's inequality implies that, with probability at least $1 - \frac{1}{n \sqrt{8 \pi \log n}} - o(1)$, there are no type I errors or type II errors.  

It may be helpful to write this in a different way.  Define the (good) event $G $ := \{For every $x \in \X$, if $\ip{x}{y} \geq \lambda + \epsilon$, then $S(x,y) \geq \mu(\lambda)$; if $\ip{x}{y} \leq \lambda - \epsilon$, then $S(x,y) \leq \mu(\lambda)$.\}  Then the theorem, combined with Markov's inequality, implies that
\[\lim_{m \rightarrow \infty} \sup_{|\X| = n} P(G^c) \leq n P(N(0,1) \geq \eta).\]

We pause to remark on how this result falls into the framework of approximate nearest neighbour search.  
    Note that $\epsilon_m$ vanishes as $m \rightarrow \infty$ and thus,
    asymptotically, the search returns precisely the documents with the desired
    level of inner product with $y$.  However, for large, but finite $m$, exact
    recovery is not expected.  Instead, there is a small interval around $\lambda$,
    and if a document falls into this interval, we cannot predict whether it will be
    correctly returned (or not returned).  Outside the interval, documents are
    returned precisely as desired with high probability.  Our main result above
    characterizes the size of this interval.  This is a version of an approximate
    nearest neighbours solution~\cite{bern1993approximate}. 

    In order to give a more traditional treatment of the nearest neighbor problem, one
    may be interested in just keeping the document with highest score, rather than keeping all documents whose scores exceed a certain threshold.  The theorem
    may be leveraged to describe the accuracy of this method. Let $x_0$ be
    the vector in $\X$ which is closest to $y$, and let $\lambda_0 = \ip{x_0}{y}$.
    While the document with highest score is not guaranteed to be $x_0$, it is
    guaranteed to have inner product with $y$ nearly as high as $\lambda_0$.  The
    easiest way to see this is to make a slight change in the definition of $G$ by
    switching $\lambda$ with $\lambda - \epsilon$.  Thus, let $\epsilon =
    \epsilon(\lambda_0, m, r, \eta)$ be defined as in Equation \eqref{eq: epsilon}
    and define the event $G' $ := \{For every $x \in \X$ , if $\ip{x}{y} \geq
    \lambda_0$, then $S(x,y) \geq \mu(\lambda_0 - \epsilon)$; if $\ip{x}{y} \leq
    \lambda_0 - 2\epsilon$, then $S(x,y) \leq \mu(\lambda_0 - \epsilon)$.\}  One can
    see, by tweaking the proof, that the theory still holds with $G$ replaced by
    $G'$.  If the event $G'$ holds, then $S(x_0, y) \geq \mu(\lambda_0 - \epsilon)$,
    and, for any $x \in \X$ with $\ip{y}{x} \leq \lambda_0 - 2 \epsilon$, one has
    $S(x,y) < \mu(\lambda_0 - \epsilon)$.  Thus, under this event, which holds with
    high probability, the document with highest score must have inner product with
    $y$ at least $\lambda_0 - 2 \epsilon$, i.e., it is an approximate nearest
    neighbour. 

In the above theorem, we focus on asymptotics in order to give a rule of thumb; in particular, note that the main term quantifying the rate at which of $\epsilon$ decreases
is  $m^{- \frac{\lambda - (2r - 1)}{2(1 + \lambda)}}$.  However, we find in numerical simulations that $m$ must be quite large to realize
the expected asymptotic behaviour.  Thus, we present the following
non-asymptotic version of this theorem, with a more complex expression to
quantify the interval, but one that may be verified in numerical simulations
even for modestly large $m$.  

For use in this theorem, we make a slight change in the definition of errors,
allowing a different value of $\epsilon$ for type I errors versus type II
errors.  
\begin{definition}[Error events]  We define two types of events:  If an $\epsilon^-$-irrelevant document is retrieved this is called a \emph{type I error}.  If an $\epsilon^+$-relevant document is not retrieved this is called a \emph{type II error}.
\end{definition}

\begin{theorem}[Main theorem, non-asymptotic version]
    \label{thm: main non-asymptotic}
Fix $y \in \sphere^{d-1}, \lambda \in (2r -1, 1)$ and $\eta > 0$.  Define $\epsilon^-$ and $\epsilon^+$ to satisfy 
\begin{equation}
    \label{eq: define epsilons}
\frac{\mu(\lambda) - \mu(\lambda - \epsilon^-)}{\sigma(\lambda - \epsilon^-)} \cdot \sqrt{m} = \eta, \qquad\quad \frac{\mu(\lambda) - \mu(\lambda + \epsilon^+)}{\sigma(\lambda + \epsilon^+)}\cdot \sqrt{m} = -\eta
\end{equation}
provided there are solutions.  

If there are solutions, then
\[\left|\sup_{|\X| = n} \frac{\E (\text{Number of type I errors}) + \E (\text{Number of type II errors})}{n}  - P(N(0,1) \geq \eta)\right| \leq \frac{1}{\sqrt{\mu(\lambda - \epsilon^-)m}}\]
\end{theorem}

For sufficiently
large $m$, Equation \eqref{eq: define epsilons} has solutions, and furthermore,
$\epsilon^-, \epsilon^+ \rightarrow 0$ (see Remark \ref{rem: epsilons} below).  

\begin{remark}[Understanding the accuracy of the Gaussian approximation]
The bound on the accuracy of the approximation $1/\sqrt{\mu(\lambda - \epsilon^-)m}$ has an intuitive meaning, which is further verified in simulations.  Indeed, note that $\mu(\lambda - \epsilon^-)m \approx \mu(\lambda) m$ is roughly the expected number of non-zero summands in the score when $\ip{x}{y} \approx \lambda$ (see Equation \eqref{eq: expected nonzeros}).  If this value is very small, one expects the score to be approximated by a discrete Poisson distribution, rather than the continuous Gaussian distribution given in the theorem.  In this sparser regime, we find that the accuracy oscillates above and below what is expected by the Gaussian approximation as $m$ is increased.  This effect is well studied in \cite{Brown2002}.
\end{remark}

%% practice.tex

\section{Practical considerations}
\label{s:practical}

Practitioners implementing our sparse mapping scheme need to consider tradeoffs
between complexity and accuracy.  In this section, we not only illustrate how
the results of the previous section inform practical design decisions that must
be made when incorporating sparse mapping into existing infrastructure, but we
also address the cost of performing the mapping itself, proposing the use of a
structured random matrix over the Gaussian random matrix used in our analysis
thus far.  Although the principal advantage of our approach is that it can be
implemented on standard search engine infrastructure, we also illustrate here
the marked improvement in complexity over an exhaustive search.  Nevertheless, we emphasize that optimizing the precise complexity-accuracy tradeoff is not our main goal in this paper; instead the goal is to map to a space utilizable by commercial search engines.

\subsection{Complexity-accuracy tradeoffs}

Previous sections point to the tradeoff between search complexity and accuracy.
Search complexity is determined by the amount of data provided to the
two processes -- indexing and searching -- undertaken by the search engine.
Accuracy is determined by the loss of fidelity in our random transform of
documents and queries.  Here we make explicit the relationship between
parameters $m$ and $r$ and their combined effect on complexity and accuracy.

\paragraph{Indexing cost}  In the indexing step, the search engine pre-processes
the documents into a data structure suitable for efficient searching.  The
indexer creates a \emph{posting list} for each term it encounters, and appends
to each posting list the pointers to documents containing that term.  In our
case, each term is an element of the $m$-dimensional output of the random
projection; documents exceeding our threshold $h$ in a dimension are added to
its respective posting list%
\footnote{%
    Readers familiar with search technology will note that a posting list can
    furthermore store the number of times the term appears in each document.
    For simplicity, we do not exploit this capacity, instead preferring to
    increase $m$, hence the number of terms we index.
}.

The cost of storing the search index is dominated by the number elements
appearing in each of the term posting lists.  Equivalently, this is the number
of unique terms in each document multiplied by the number of documents.  After
our transformation, sparsity estimates~\eqref{eq: Gaussian tail} indicate that
indexing the entire corpus, size $|\X|=n$, has an expected storage cost of
\begin{equation}
    \bigO(nk) = \bigO\left( \frac{n m^{1-r}}{\sqrt{r \log m}} \right).
\end{equation}
That is, $r$ determines the rate at which indexing complexity increases with $m$.
Figure~\ref{f:sparsity-gaussian} is a plot of the relationship between $m$ and $\E
k$ for various $r$. 

\begin{figure}
    \centering
    \begin{subfigure}[b]{0.44\textwidth}
        \includegraphics[width=\textwidth]{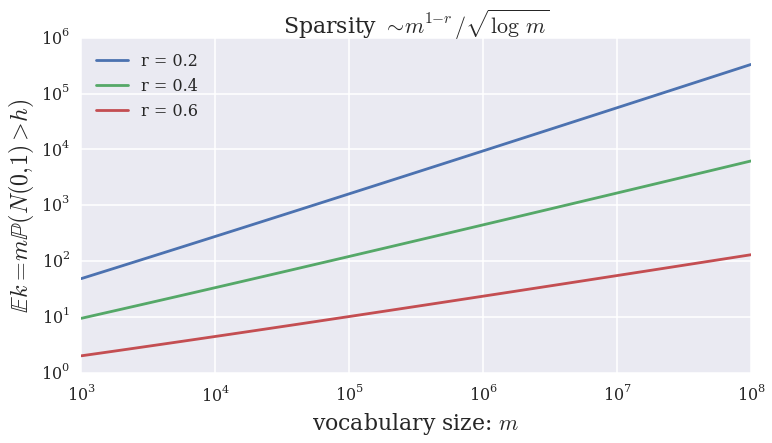}
        \caption{Gaussian mapping (calculated)}
        \label{f:sparsity-gaussian}
    \end{subfigure}
    \hspace{0.1\textwidth}
    \begin{subfigure}[b]{0.44\textwidth}
        \includegraphics[width=\textwidth]{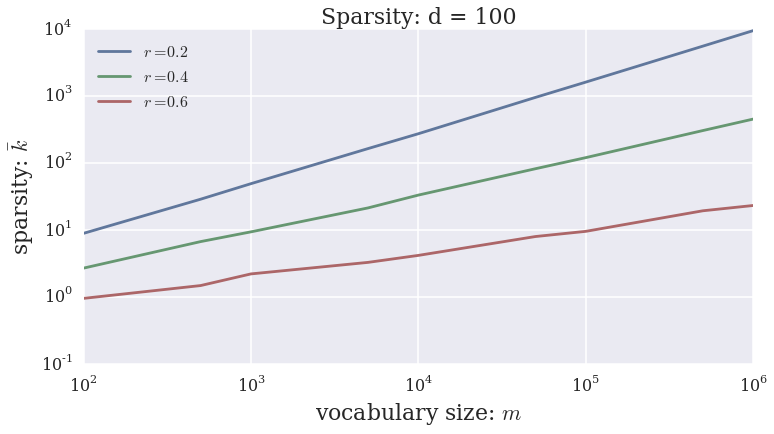}
        \caption{Structured mapping (simulated)}
        \label{f:sparsity-simulated}
    \end{subfigure}
    \caption{Expected posting list length, or sparsity, showing the effect of
        $r$ on the rate of sparsity increase with $m$.  Note that since $r>0$,
        the expected number of posting lists in which each document appears is
        in sublinear in $m$.}
    \label{f:sparsity}
\end{figure}

\paragraph{Search cost}  Sparsity in the transformed space also determines
search time complexity.  We first bound worst-case complexity.  The search engine must expect to examine the $k$
posting lists indicated by the query vector (in transformed space).  Each of these lists has length bounded by $n$, and hence each search must examine
\[ \E \bigO(nk) = \bigO\left(n\frac{m^{1 - 2r}}{\sqrt{2r \log m}}\right)\]
documents.  However, in practice one would not expect each mapped vector in the corpus to have precisely the same support as the mapped query vector, and thus each list size would be much smaller.  

We now give a rough average-case complexity. In practice, we find most data sets are clustered, i.e., there are a cluster of corpus elements close to the query, and the rest are nearly orthogonal to the query.  Indeed, in image search, most images in a data base have nothing to do with the query, and unrelated (random), high-dimensional, vectors tend to be nearly orthogonal\footnote{The inner product between random high dimensional vectors concentrates very close to zero \cite{LT}.}.  Thus, as an approximation, consider the case when $n_1$ of the corpus elements are near to the query, and $n - n_1$ corpus elements are orthogonal to the query.  Further, note that when the query vector and the corpus vector are orthogonal, the random mappings are independent.  Thus, for each of the $n - n_1$ corpus elements orthogonal to the query, the probability that this element contributes to a posting list is bounded by $\E \bigO(k/m)$.  It follows that each posting list has an expected length bounded by $\E \bigO((n - n_1) k/m) + \E \bigO(n_1)$.  If $n \gg n_1 m/k$, the first term dominates.  Then, since there are $k$ posting lists, each search must examine an expected
\begin{equation}
    \E \bigO(nk^2/m) = \bigO\left( \frac{n m^{1-2r}}{r \log m} \right)
\end{equation}
documents.  Here is where we see an advantage over the exhaustive comparison of
the query against each of $n$ documents, an $\bigO(nd)$ procedure.

Finally, to speed up search, we suggest using a larger value of $h$ to generate
queries than that used to index the corpus.  This makes query vectors sparser than
the $\E k$ calculated above.  This modification leverages engineering design choices
consequential to an asymmetry in text-based search, where queries tend to have
many fewer terms than the text documents they retrieve.  For simplicity, we do not address this scenario in our section on theory (although we believe it would be straightforward to adjust our theory to this setting).  However, our search
experiments of Section~\ref{s:experiments} use larger $h$ for search query
transformation than for indexing with no loss of fidelity, with significant
improvements in search-time performance.

Of course, the costs of indexing and search are determined not only by
transformed document sparsity, but also by the cost of doing the transformation.
In the case of our transformation by a dense Gaussian matrix, this is the cost
of a matrix multiplication, $\bigO(md)$.  We can make this cost sublinear in $m$
if we use a structured random matrix to transform our corpus vectors:  see
Section~\ref{s:structured}, below.

\paragraph{Accuracy}  We focus our discussion of accuracy on the rate at which
we erroneously return documents that do not match the query.  The essential
result is that the interval $[\lambda-\epsilon^-, \lambda+\epsilon^+]$ shrinks
as $m$ increases at a rate that depends on $r$.  Although given implicity
by~\eqref{eq: define epsilons}, the relationship between this confidence
interval and $m$ is approximately $\epsilon^++\epsilon^- \sim
\bigO(\sqrt{m^{r-1}})$.  (Up to logarithmic factors, this heuristic matches the asymptotic theoretical rate
of decay of $\epsilon$ given in Theorem \ref{thm: main} for $\lambda$
approaching 1.) 
Figure~\ref{f:epsilon} gives an example relationship between
$\epsilon^\pm$ and $m$ for $\lambda=0.9$.  Naturally, accuracy in this sense
increases with $m$, although contrary to cost estimates, does so more rapidly
with decreasing $r$.
\begin{figure}
    \centering
    \includegraphics[width=0.8\textwidth]{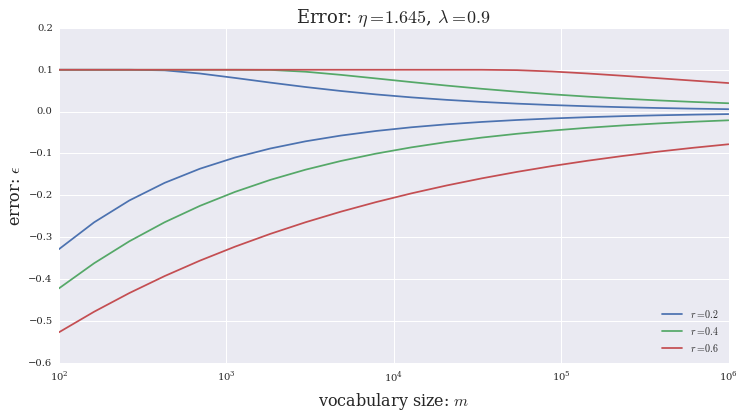}
    \caption{Relationship between accuracy and $m$ for various parameters $r$.
    Two curves for each value of $r$ are shown, the upper and lower curves
    representing $\epsilon^+$ and $\epsilon^-$, respectively.  The choice
    $\eta=1.645$ corresponds to a 5\% probability that we encounter a type I or
    type II error.}
    \label{f:epsilon}
\end{figure}

\subsection{Fast mappings using structured matrices}
\label{s:structured}

Because a Gaussian random matrix $A \in \real^{m\times d}$ is dense and
unstructured, transforming each corpus document and each query is $\bigO(md)$.
To speed this up, we take advantage of the growing literature on
structured random mappings, which suggests that structured random
matrices, which allow fast transforms, behave similarly to Gaussian
matrices.  In particular, \cite{krahmer2011new} shows that the metric-preserving
property of the Johnson-Lindenstrauss Lemma can be achieved via multiplication
by a random diagonal matrix followed by a (fast) discrete Fourier transform.
Thus, we try a similar fast, structured mapping, and show through numerical
simulations that it behaves much the same as a Gaussian matrix, provided that
$m$ is large enough, and the documents are not overly sparse.

The structured mapping we propose is in two steps.  First, we apply a random
but fast linear mapping of our vector $x \in \sphere^{d-1}$ to an intermediary $u \in
\real^m$.  Then, we perform a discrete cosine transform on $u$, thresholding
each element of the output by $h>0$, just as we did for the Gaussian transform.

In our experiments, the specific transform $D$ we choose to expand $u = Dx$
creates $m/d$ copies of $x$, then randomly changes the sign of each element.
That is, we may represent this $\bigO(m)$ operation as multiplication by matrix
\begin{equation}
    D = \begin{bmatrix} D_1 \\ D_2 \\ \vdots \\ D_{(m/d)} \end{bmatrix},
\end{equation}
where each block $D_r \in \real^{d\times d}$ is a diagonal matrix of random $\pm
1$.  We have restricted $m, d$ such that $d$ is a divisor of $m$. 

To intermediary $u$, we apply a normalized type 2 discrete cosine transform
(DCT-II), which we can represent as $v = Fu$, or specifically,
\begin{align}
    v_i &= c_i\sqrt{d/m} \sum_{j=1}^{m} u_j \cos\left(\frac{\pi (i-1)(2j-1)}{2m}\right),
    \quad i=1 \ldots m \nonumber \\
    c_i &= \begin{cases} 1, & i=1 \\ \sqrt{2}, & \text{otherwise} \end{cases}
    \label{eq:dct2}
\end{align}
This is an $\bigO(m\log m)$ transform, just as is the complex discrete Fourier
transform.  The normalization factor we choose for $F$ in~\eqref{eq:dct2}
ensures that $\E \norm{FDx}^2 = m$, just as $\E \norm{Ax}^2 = m$ where $A \in
\real^{m \times d}$ is the Gaussian matrix we chose previously.  With this
choice, we can use the same threshold $h$ that we used in the Gaussian analysis.

Computer experiments suggest that our proposed structured mapping approximates
the behaviour of the Gaussian mapping that we prove in our main theorems.
Figure~\ref{f:sparsity-simulated} illustrates not only that sparsity grows at a
rate comparable to the $\bigO(m^{1-r}/\sqrt{\log m})$ rate realized by the Gaussian
mapping, but bears the same absolute values.

Figures~\ref{f:accuracy-simulated} demonstrate similar behaviour of our Gaussian
and structured mappings in the model case where we have a corpus of a single
document.  We query $\X = \{x\}$ using vector $y$, setting
$\ip{x}{y}=\lambda-\epsilon^-$ to examine type I error and
$\ip{x}{y}=\lambda+\epsilon^+$ when examining type II error.  Queries are
iterated over many random \emph{transforms} (as opposed to many random $x, y$)
to measure error rates.

Note that choosing $d=2$ in simulating the Gaussian mapping gives the same
result as larger $d$:  all that we require for this mapping is that each mapped
element $\ip{a_i}{x} \sim N(0,1)$.  For testing the structured mapping, we
select dense $x,y$ of modest dimension $d=100$.

\begin{figure}
    \centering
    \begin{subfigure}[b]{0.44\textwidth}
        \includegraphics[width=\textwidth]{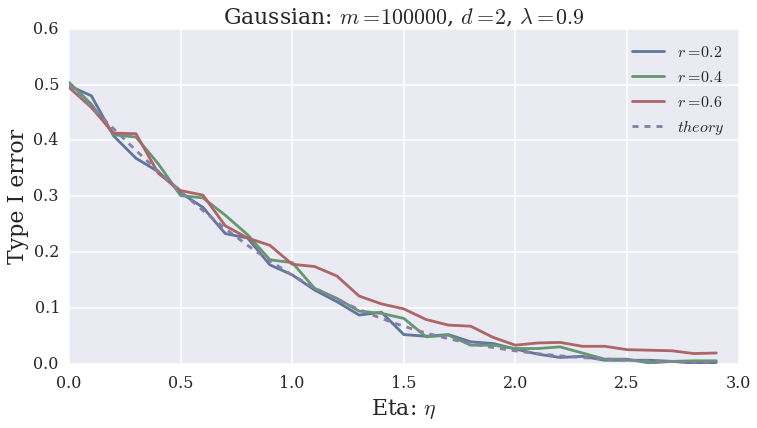}
        \caption{Gaussian type I error}
        \label{f:gaussian-typeI}
    \end{subfigure}
    \hspace{0.1\textwidth}
    \begin{subfigure}[b]{0.44\textwidth}
        \includegraphics[width=\textwidth]{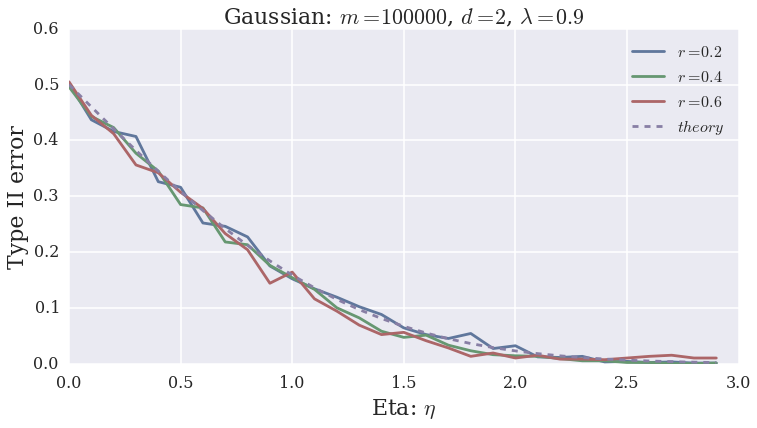}
        \caption{Gaussian type II error}
        \label{f:gaussian-typeII}
    \end{subfigure}

    \vspace{0.3in}
    \begin{subfigure}[b]{0.44\textwidth}
        \includegraphics[width=\textwidth]{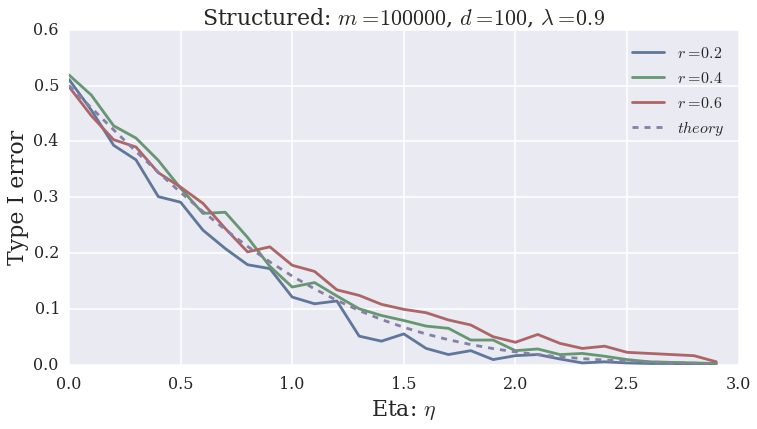}
        \caption{Structured type I error}
        \label{f:structured-typeI}
    \end{subfigure}
    \hspace{0.1\textwidth}
    \begin{subfigure}[b]{0.44\textwidth}
        \includegraphics[width=\textwidth]{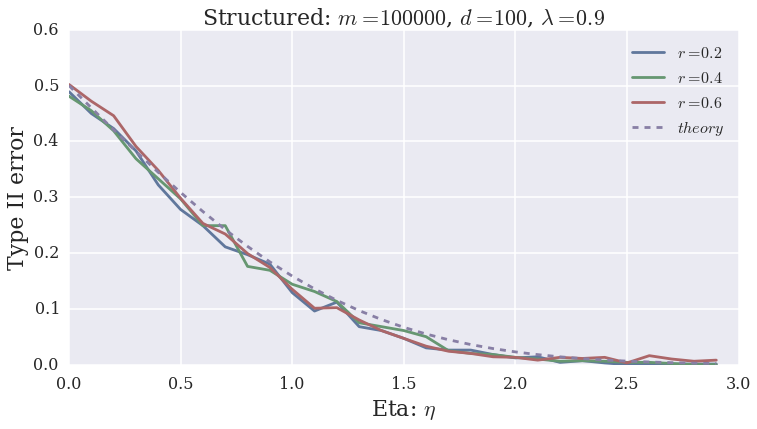}
        \caption{Structured type II error}
        \label{f:structured-typeII}
    \end{subfigure}

    \vspace{0.3in}
    \begin{subfigure}[b]{0.44\textwidth}
        \includegraphics[width=\textwidth]{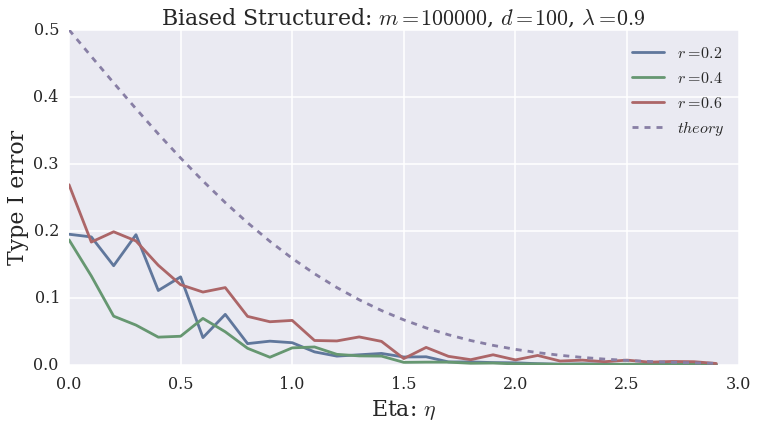}
        \caption{Biased structured type I error}
        \label{f:biased-typeI}
    \end{subfigure}
    \hspace{0.1\textwidth}
    \begin{subfigure}[b]{0.44\textwidth}
        \includegraphics[width=\textwidth]{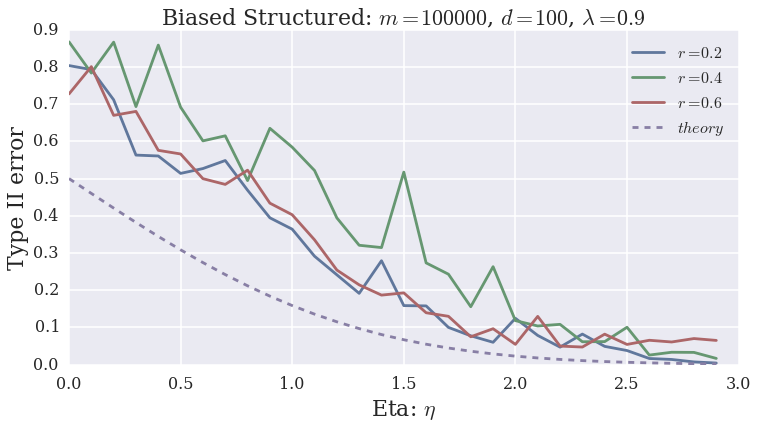}
        \caption{Biased structured type II error}
        \label{f:biased-typeII}
    \end{subfigure}
    \caption{Comparison of Gaussian transformation supported by theory and
        simulated results of the structured mapping for a random 1-element corpus.}
    \label{f:accuracy-simulated}
\end{figure}

We do not believe that the particular $F, D$ we choose to construct our
structured random mapping are the only possibilities.  In particular, we
imagine that the speed of at least the Fourier transform step can be increased by
leveraging recent work in randomized Fourier transforms that realize $\bigO(k
\log m \log(m/d))$ complexity of determining the $k$ larges entries of the Fourier transform~\cite{hassanieh2012simple}.  

However, we caution the reader that even if normalized correctly, the expansion
$u=Dx$ cannot be arbitrary: In one of our early attempts, we chose $D\in
\real^{m\times d}$ a matrix with exactly one random $\pm 1$ in each column and
at most one non-zero in each row.  Although the sparsity of the resulting map
(not illustrated) is comparable to the sparsity of the Gaussian map, the
asymmetry between simulated type I and type II errors shown in
Figures~\ref{f:biased-typeI} and~\ref{f:biased-typeII} indicate output biased
towards moving relevant documents apart from each other.  We conjecture that for
an unbiased transform, the vector we input to the Fourier transform must be
dense with elements having mean zero.  However, a study of the class of
structured random mappings which approximate the behaviour proven in our main
theorems remains, for the time being, future work.

%% experiments.tex

\section{Experiments}
\label{s:experiments}

We now evaluate the performance of the random mapping approach using two
different datasets:  search based on the color of Wikipedia images, and Dow
Industrial market data based on closing value percentage differences going back
to the inception of this index.  Searches take a seed item as a query, and
attempt to find other items having similar features. We note that our
two examples are from remarkably different areas of study:  all that is required 
is that corpus documents be represented by vectors in $\real^d$. In both cases, 
we evaluate recall and precision as well as ranking of top search results. Below, \textit{relevant} is the set of all relevant documents in the corpus and \textit{retrieved} is the set of all retrieved documents in the corpus.  Precision and recall scores are calculated as:
\begin{equation}
    precision = \frac{|relevant \cap retreived|}{|retrieved|}, \quad recall =
    \frac{|relevant \cap retrieved|}{|relevant|}
\end{equation}
\subsection{ImageCLEF Wikipedia image corpus}

For its intuitive evaluation and ubiquity, our first test is a search of images
of common color in the ImageCLEF 2010 Wikipedia
Collection~\cite{muller2010experimental}.  We process each image by transforming
images to HSV colorspace, and binning the pixels of each image into histograms.
Two example searches are shown in Figure~\ref{f:image-search}.  For these
searches, we use other images as the queries, and so the first step in each
search is to extract the color histogram from the query image before
transforming it.
\begin{figure}
    \centering
    \begin{subfigure}[b] {0.8\textwidth}
        \includegraphics[width=\textwidth]{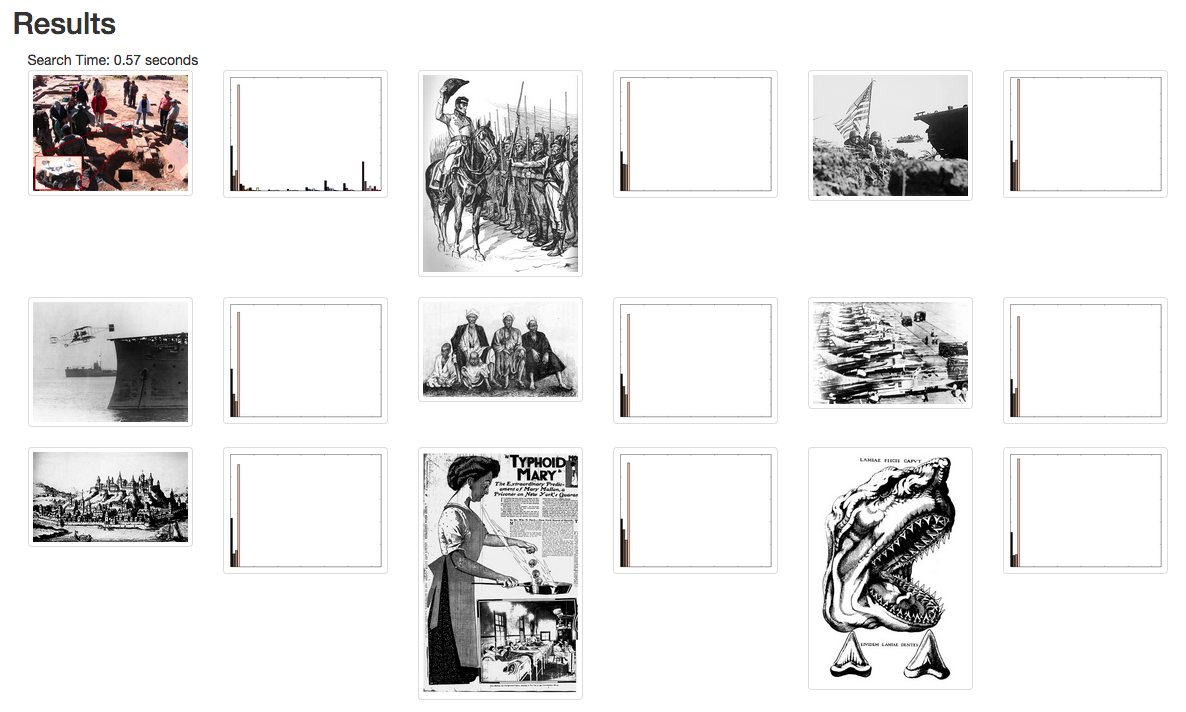}
        \caption{Mediocre result}
        \label{f:image-poor}
    \end{subfigure}
    \vspace{0.3in}
    \begin{subfigure}[b] {0.8\textwidth}
        \includegraphics[width=\textwidth]{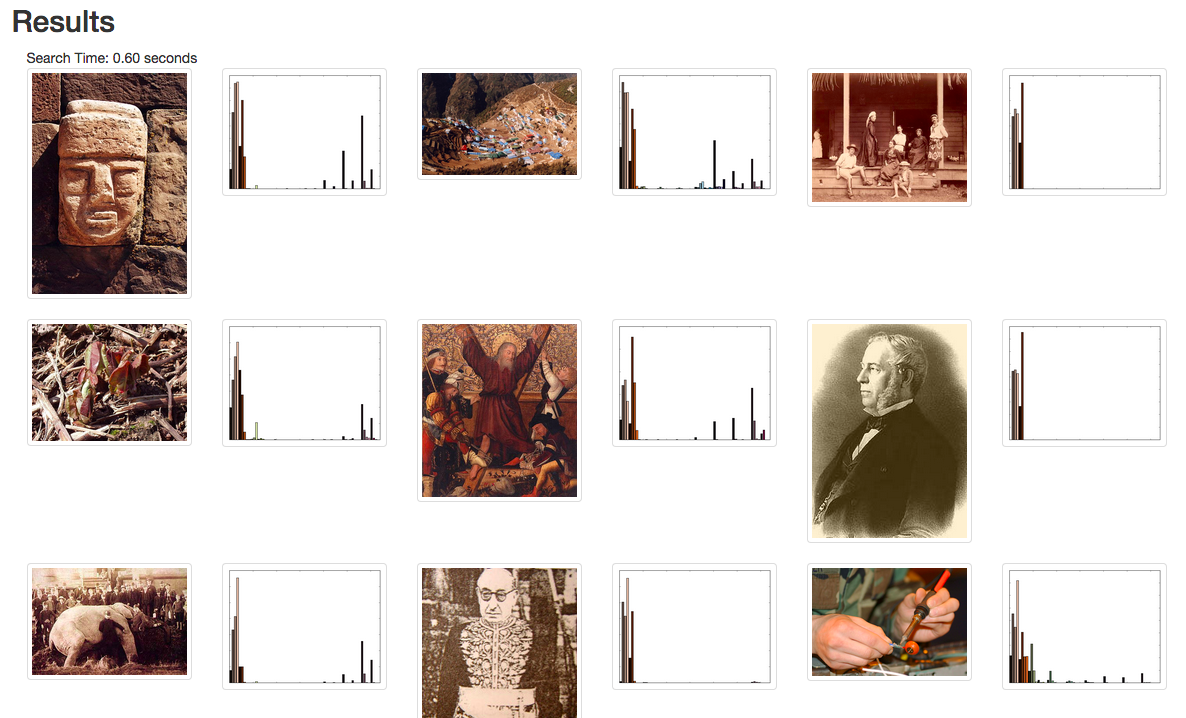}
        \caption{Good result}
        \label{f:image-good}
    \end{subfigure}
    \caption{Example searches for images of similar color over the Wikipedia
        ImageCLEF corpus.  The top left is the query.  Note that even in the case of the visually mediocre
        result, the match of the histograms is good.  The medicore result is
        mediocre because of the lack of good color resolution in the HSV bins
        for this query.}
    \label{f:image-search}
\end{figure}

Results for histograms of 128 color bins are shown.  We achieve good results
with corpora represented by 32 and 64 bins.  However, as shown by the
search-time comparison in Figure~\ref{f:search-time}, our method does not
penalize large document vectors, and we choose the higher-resolution histograms.
We also choose our structured random embedding over the Gaussian transformation
at no loss in functionality.  Note that the histograms accompany each image in
Figure~\ref{f:image-search} so that errors due to our transform may be
distinguished from quantization error due to color histogram binning.
\begin{figure}
    \centering
    \includegraphics[width=0.8\textwidth]{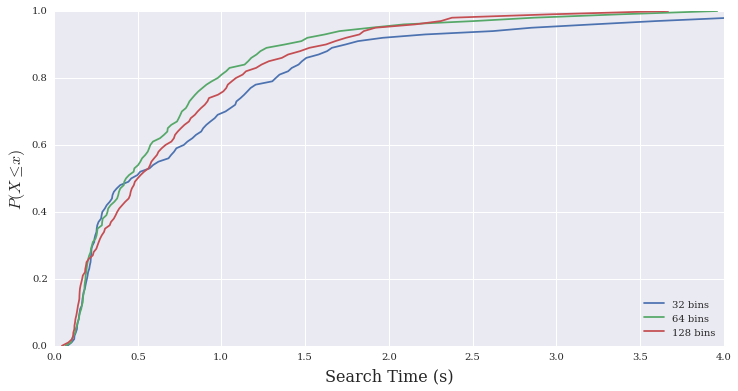}
    \caption{Empirical cumulative distributions of search times for each of
        the three color histogram binning schemes, namely, $d=32, 64, 128$.
        Search times have similar statistics, independent of $d$.}
    \label{f:search-time}
\end{figure}

Figure~\ref{f:pr-image} gives a more quantitative representation of performance
as the precision-recall curve.  To generate each point on these curves, we
fix $\ip{x}{y} = \lambda$ our threshold for relevant documents, and $S(x,y)
= \mu(\lambda)$ our threshold for returned documents.  Error bars representing
the standard error of precision and recall are generated
for each point by sampling over may query images.
\begin{figure}
    \centering
    \includegraphics[width=0.8\textwidth]{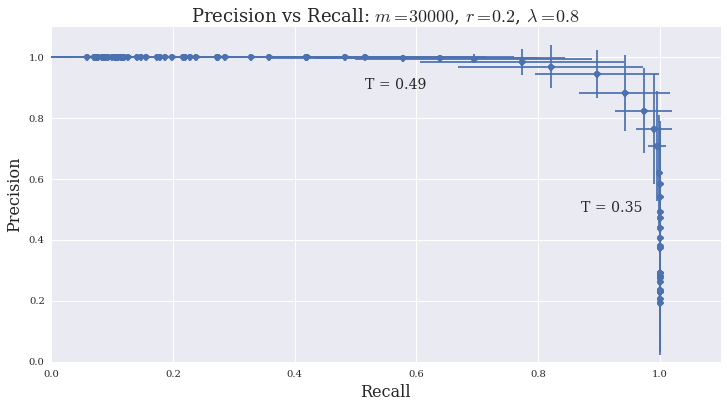}
    \caption{Precision-Recall curve for 128-bin histograms showing the
        performance of the sparse transform parameterized by $T$, the threshold
        for relevant results.  Error bars represent a range of performance over
        randomly sampled query images.}
    \label{f:pr-image}
\end{figure}
Naturally, precision increases with $\lambda$, but most importantly, the area
under the mean precision-recall curve is nearly 1, showing that our method tends
to preserve the ordering of ranked results.

\subsection{Dow Industrial corpus}

Inspired by the finance literature employing nearest neighbour estimates for
forecasting markets, the second test of our method draws on the Dow Jones Industrial
Average~\cite{centralbankbrazil}.  Each vector of this
corpus is a 10-element vector of relative differences between closing index
values for each of preceeding 5 and succeeding 5 trading days.  ``Bullish''
elements are characterized by large positive differences, while ``bearish''
differences are characterized by large negative differences.
A search, thus, takes a single day as query and finds
days with similar time-local trading patterns.  Example results of two searches,
queried by each of a bearish and bullish day, are shown in
Table~\ref{t:dow-search1} and Table~\ref{t:dow-search2}.
\begin{table}
\footnotesize
\centering
\begin{minipage}{0.48\textwidth}
\begin{tabular}{lrrc}
\hline\hline
Date & \% Change & True score & Vector \vspace{1 mm} \\\hline \\

1907-03-14 & -8.289 & 1.0 & \includegraphics[height=0.2in, width=1.0in]{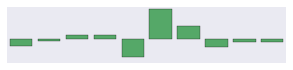}\\
1957-05-28 & -0.298 & 0.916 & \includegraphics[height=0.2in, width=1.0in]{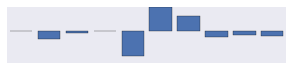}\\
1945-08-08 & 0.173 & 0.934 & \includegraphics[height=0.2in, width=1.0in]{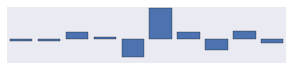}\\
1943-11-18 & 0.422 & 0.917 & \includegraphics[height=0.2in, width=1.0in]{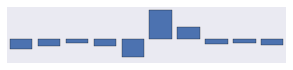}\\
1972-09-26 & 0.089 & 0.935 & \includegraphics[height=0.2in, width=1.0in]{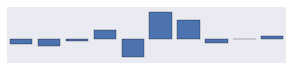}\\
1898-01-05 & 1.268 & 0.917 & \includegraphics[height=0.2in, width=1.0in]{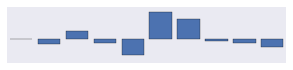}\\
1971-11-02 & 0.257 & 0.863 & \includegraphics[height=0.2in, width=1.0in]{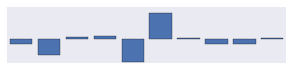}\\
1901-05-09 & -6.051 & 0.881 & \includegraphics[height=0.2in, width=1.0in]{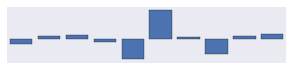}\\
1975-03-25 & 0.6 & 0.843 & \includegraphics[height=0.2in, width=1.0in]{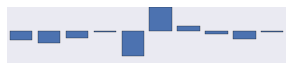}\\
1985-08-19 & -0.017 & 0.868 & \includegraphics[height=0.2in, width=1.0in]{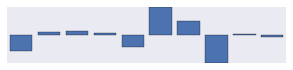}\\
1929-10-29 & -11.729 & 0.865 & \includegraphics[height=0.2in, width=1.0in]{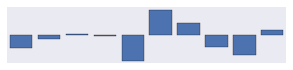}\\
1896-07-21 & 0.59 & 0.844 & \includegraphics[height=0.2in, width=1.0in]{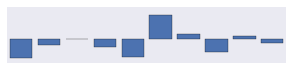}\\
1979-09-19 & 0.263 & 0.858 & \includegraphics[height=0.2in, width=1.0in]{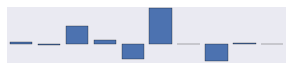}\\
1968-01-31 & -0.477 & 0.905 & \includegraphics[height=0.2in, width=1.0in]{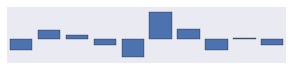}\\
\\
\hline\hline
\end{tabular}
\caption{Query is the Dow bust of 1907-03-14 (top row).}
\label{t:dow-search1}
\end{minipage}
\hfill
\begin{minipage}{0.48\textwidth}
\centering
\begin{tabular}{lrrc}
\hline\hline
 Date & \% Change & True Score & Vector \vspace{1 mm} \\\hline \\
2008-10-28 & 10.878 & 1.0 & \includegraphics[height=0.2in, width=1.0in]{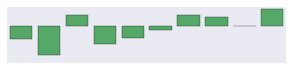}\\
1980-12-12 & 0.958 & 0.939 & \includegraphics[height=0.2in, width=1.0in]{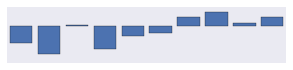}\\
1898-07-18 & -0.191 & 0.873 & \includegraphics[height=0.2in, width=1.0in]{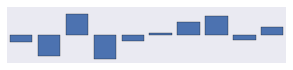}\\
1930-02-26 & 2.382 & 0.901 & \includegraphics[height=0.2in, width=1.0in]{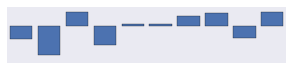}\\
1927-03-07 & -0.383 & 0.883 & \includegraphics[height=0.2in, width=1.0in]{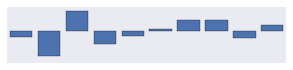}\\
1909-11-11 & -0.01 & 0.822 & \includegraphics[height=0.2in, width=1.0in]{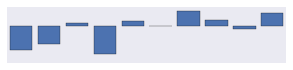}\\
1962-09-28 & 0.847 & 0.872 & \includegraphics[height=0.2in, width=1.0in]{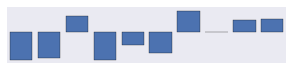}\\
1939-02-11 & 0.647 & 0.818 & \includegraphics[height=0.2in, width=1.0in]{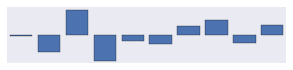}\\
1900-09-25 & 0.321 & 0.839 & \includegraphics[height=0.2in, width=1.0in]{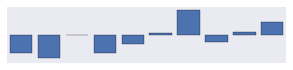}\\
1913-05-16 & 0.408 & 0.889 & \includegraphics[height=0.2in, width=1.0in]{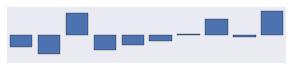}\\
1973-05-01 & -0.024 & 0.76 & \includegraphics[height=0.2in, width=1.0in]{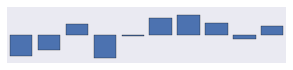}\\
1938-04-30 & -0.34 & 0.884 & \includegraphics[height=0.2in, width=1.0in]{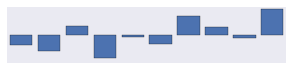}\\
1980-04-22 & 4.047 & 0.869 & \includegraphics[height=0.2in, width=1.0in]{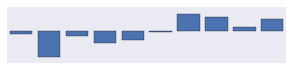}\\
1979-02-06 & -0.137 & 0.9 & \includegraphics[height=0.2in, width=1.0in]{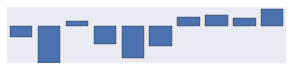}\\
\\
\hline\hline
\end{tabular}
\caption{Query is the Dow boom of 2008-10-28 (top row).}
\label{t:dow-search2}
\end{minipage}
\end{table}

This second example has two appealing features:  First, it is \emph{not} an image
corpus, illustrating that our method is not an image search method (though it
may find application there), but is a search for any data well-represented by
vectors in $\real^d$.  Second, as shown in Figure~\ref{f:stats-compare}, the
statistics of these Dow data differ significantly from those of the image data.
Despite these different statistics, the precision-recall curve for the Dow
Industrial data, Figure~\ref{f:pr-dow}, shows the ranked search performance to be 
comparable to that for the image data.
\begin{figure}
    \centering
    \includegraphics[width=0.8\textwidth]{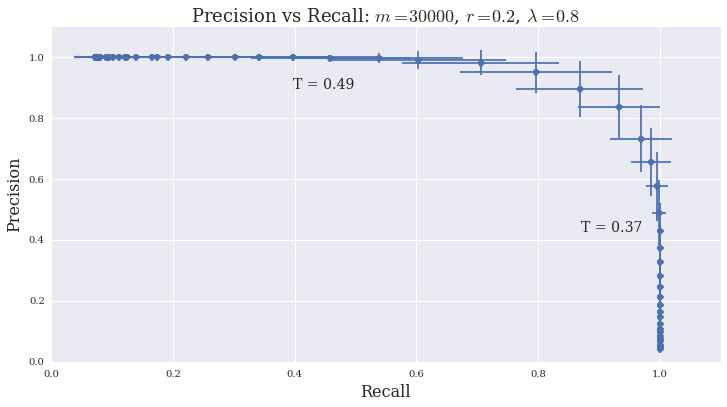}
    \caption{Precision-recall curve for the Dow data parameterized by $T$,
        the inner product threshold for relevant results.}
    \label{f:pr-dow}
\end{figure}
\begin{figure}
    \centering
    \begin{subfigure}[b] {0.45\textwidth}
        \includegraphics[width=\textwidth]{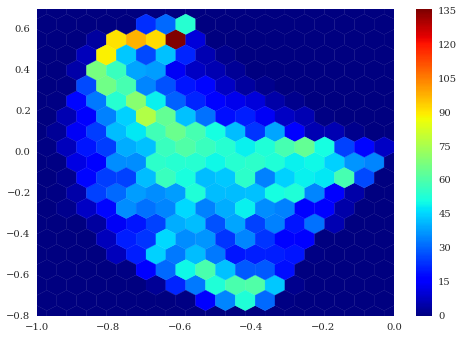}
        \caption{Wikipedia ImageCLEF 128-bin color histograms.}
    \end{subfigure}
    \begin{subfigure}[b] {0.45\textwidth}
        \includegraphics[width=\textwidth]{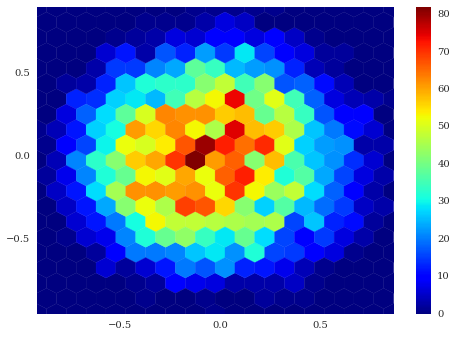}
        \caption{Dow Industrial daily differences, 10-day window.}
    \end{subfigure}
    \caption{Comparison of the statistics of the Wikipedia ImageCLEF and Dow
        Jones Industrial corpora.  Each plot shows the density of the projection 
        of the corpora onto their first two singular vectors.}
    \label{f:stats-compare}
\end{figure}

\subsection{Online resources}

See~\url{https://gitlab.com/dgpr-sparse-search/code} for code for the
simulations in Section~\ref{s:practical} and the search demonstrations.  The
code for this project is written in Python, with simulations appearing as
IPython notebooks.  Whereever possible, custom code is avoided in favour of
off-the-shelf open-source projects.  Search examples use the Django web
framework, and are powered by the Whoosh search package.  As evidenced by the
long search times (median $\approx 500\text{ ms}$ for our set of 270K images),
Whoosh is not the fastest off-the-shelf search engine.  Rather, we select it for
its ease of configuration and structural similarity to compiled off-the-shelf
search engines such as Apache Lucene and ElasticSearch.

%% endmatter.tex

\bibliographystyle{plain}
\bibliography{sparse-search}

%% appendix.tex

\appendix
\section{Proofs}
\subsection{Proof of Lemma \ref{lem: Gaussian score}}\label{ssec: gaussian score proof} 
We will need the following non-asymptotic version of the Central Limit Theorem~\cite{berry-esseen-constant}.

\begin{theorem}[Berry-Esseen Central Limit Theorem]
\label{thm: berry-esseen}
Let $z_1, z_2, \hdots, z_m$ be independent, identically distributed, mean-zero, random variables satisfying $\E z_i^2 = \sigma^2$.  Set
\[S_m = \frac{\sum_{i = 1}^m z_i}{\sigma \sqrt{m}} \quad \text{and} \quad \rho = \E |z_1|^3.\]
Let $F_m(t)$ be the cumulative distribution function of $S_m$.  Then for all $t$ and $m$,
\begin{equation}\label{eq: berry-esseen}
|F_m(t) - \Phi(t)| \leq C_0\frac{\rho}{\sigma^3 \sqrt{m}}
\end{equation}
where $C_0 = 0.4748$.
\end{theorem}

We may leverage Theorem~\ref{thm: berry-esseen} to characterize the rate at
which $\tilde{S}(x,y)$ converges to a standard normal random variable.
\begin{lemma}\label{lem: non-asymptotic normal}
 Let everything be as in Lemma \ref{lem: Gaussian score}, but with no
 restriction on $\lambda$ (aside from $\lambda \in [-1, 1]$). Note that the
 distribution of $\tilde{S}(x,y)$ depends only on $m$ and $\lambda$, and define
 $F_{m, \lambda}(t) := P(\tilde{S}(x,y) < t)$ to be the corresponding cumulative
 distribution function.  Then, \begin{equation}\label{eq: non-asymptotic normal}
|F_{m, \lambda}(t) - \Phi(t)| \leq \frac{1}{\sqrt{ \mu(\lambda) m}} \qquad \text{for all} \quad t \in \R.
\end{equation}
\end{lemma} 

\begin{proof}
We apply Theorem \ref{thm: berry-esseen} to $\tilde{S}(x,y)$.  Thus, let $z_i =
\one_{[w_i > h]} \one_{[v_i > h]} - \mu(\lambda)$, and note that the normalized
score satisfies
\[\tilde{S}(x,y) = \frac{\sum_{i = 1}^m z_i}{\sigma(\lambda) \sqrt{m}}.\]
It is not hard to bound $\rho$ as follows 
\[\rho := \E|z_1|^3 \leq \mu(\lambda)(1-\mu(\lambda)) = \sigma^2(\lambda).\] 
Thus, the right-hand side of the Berry-Esseen bound Equation~\eqref{eq:
berry-esseen} is less than $1/(\sigma(\lambda) \sqrt{m})$.  Further,
$\mu(\lambda) \leq P(N(0,1) > h) \leq P(N(0,1) > 0) = 1/2$ and thus,
$\sigma(\lambda) = \sqrt{\mu(\lambda)(1 - \mu(\lambda))} \geq
\sqrt{\mu(\lambda)/2}$.

Then Theorem~\ref{thm: berry-esseen} implies that for all $t$ and $m$, the
cumulative distribution function of $\tilde{S}(x,y)$ satisfies \[|F_{m,
\lambda}(t) - \Phi(t)| \leq \frac{0.4748 \sqrt{2}}{ \sqrt{ \mu(\lambda) m}} \leq
\frac{1}{\sqrt{\mu(\lambda) m}}.\]
\end{proof}

This result quickly translates into a proof of Lemma~\ref{lem: Gaussian score}.
\begin{proof}[Proof of Lemma~\ref{lem: Gaussian score}]
Case 1 has already been proven above.  Case 2 follows from from Lemma \ref{lem:
non-asymptotic normal}; we only need to show that $\frac{1}{\sqrt{\mu(\lambda)
m}} \rightarrow 0$ for $\lambda \in (2r-1, 1)$.  This follows from Equation
\eqref{eq: expected nonzeros}.  \end{proof} 

\subsection{Proof of non-asymptotic main theorem: Theorem~\ref{thm: main non-asymptotic}}

We begin with a few lemmas, which determine the behaviour of the image search
procedure non-asymptotically, and in the simple case when $|\X| = 1$.

\begin{lemma}\label{lem: non-asymptotic return x}
Fix $x, y \in \sphere^{d-1}$, and let $\lambda' := \ip{x}{y}$.  Fix $\lambda \in [-1,1]$.

The probability that $S(x,y) \geq \mu(\lambda)$ (i.e., the event that we return $x$) satisfies the following bound:
\begin{equation}\label{eq: non-asymptotic return x}
\left|P(S(x,y) \geq \mu(\lambda)) - P\left(N(0,1) > \frac{\mu(\lambda) - \mu(\lambda')}{\sigma(\lambda')}\cdot \sqrt{m} \right)\right| \leq \frac{1}{ \sqrt{ \mu(\lambda') m}}.
\end{equation}
\end{lemma}

\begin{proof}
\[P(S(x,y) \geq \mu(\lambda)) = P\left(\tilde{S}(x,y) \geq \frac{\mu(\lambda) - \mu(\lambda')}{\sigma(\lambda')} \cdot \sqrt{m}\right).\]
Now apply Equation \eqref{eq: non-asymptotic normal} of Lemma \ref{lem: Gaussian score} to complete the proof.
\end{proof}
We may synthesize this result to give an interval around $\lambda$ outside of which one would expect to return (and not return) precisely the desired documents.

\begin{lemma}\label{lem: non-asymptotic G}
Fix $y \in \sphere^{d-1}, \lambda \in [-1,1]$ and $\eta > 0$.  Define $\epsilon^-$ and $\epsilon^+$ as in Theorem~\ref{thm: main non-asymptotic}.
Consider the (good) event $G_x := $\{if $\ip{x}{y} \geq \lambda + \epsilon^+$, then $S(x,y) \geq \mu(\lambda)$; if $\ip{x}{y} \leq \lambda - \epsilon^-$, then $S(x,y) \leq \mu(\lambda)$.\}
Then
\[|\sup_{x} P(G_x^c) - P(N(0,1) \geq \eta)| \leq \frac{1}{\sqrt{\mu(\lambda - \epsilon^-)m}}.\]
\end{lemma} 
\begin{proof}
We begin with the following observation about the behaviour of the score.  Fix $x$ and $x'$ with $\lambda = \ip{x}{y}$ and $\lambda' = \ip{x'}{y}$.  Suppose that $\lambda' \leq \lambda$.  Then $S(x,y)$ probabilistically dominates $S(x',y)$, i.e., for any $t \in \R$
\begin{equation}\label{eq: prob dom}
P(S(x', y) \geq t) \leq P(S(x,y) \geq t).
\end{equation}
The above is a simple consequence of the fact that $m S(x', y)$ and $m S(x,y)$ are both binomially distributed, with respective means of $m\mu(\lambda')$ and $m\mu(\lambda)$, and $\mu(\lambda') \leq \mu(\lambda)$.  

We will use this observation below, but first fix $x$ with $\ip{x}{y} = \lambda - \epsilon^-$.  Then, by Lemma \ref{lem: non-asymptotic return x},
\begin{equation}\label{eq: lower bound for Gxc}
P(G_x^c) \geq P(N(0,1) \geq \eta) - \frac{1}{\sqrt{\mu(\lambda - \epsilon^-)m}} \qquad \text{and} \qquad P(G_x^c) \leq P(N(0,1) \geq \eta) + \frac{1}{\sqrt{\mu(\lambda - \epsilon^-)m}}.
\end{equation}
We will now show that for any other vector $x'$
\[P(G_{x'}^c) \leq P(N(0,1) \geq \eta) + \frac{1}{\sqrt{\mu(\lambda - \epsilon^-)m}},\]
which will complete the proof.  Thus, fix $x'$ with $\ip{x'}{y} = \lambda'$.  

There are four cases to consider:

\noindent{\bf Case 1:  $\lambda' \leq \lambda - \epsilon^-$.} 
Since $S(x,y)$ probabilistically dominates $S(x',y)$, we have
\[P(G_{x'}^c) \leq P(G_x^c)\leq P(N(0,1) \geq \eta) + \frac{1}{\sqrt{\mu(\lambda - \epsilon^-)m}}\]
as desired.

\noindent{\bf Case 2:  $\lambda' \in (\lambda - \epsilon^-, \lambda + \epsilon^+)$.}
Then, clearly, 
\[P(G_{x'}^c) = 0.\]

\noindent{\bf Case 3:  $\lambda' = \lambda + \epsilon^+.$}
If $\lambda' > 1$, then clearly $P(G_{x'}^c) = 0$ since $\ip{x}{y}$ is always bounded by 1.  Thus, suppose $\lambda' \leq 1$.  Then, by Lemma \ref{lem: non-asymptotic return x},
\[P(G_{x'}^c) \leq P(N(0,1) \geq \eta) + \frac{1}{\sqrt{\mu(\lambda + \epsilon^+)m}} \leq P(N(0,1) \geq \eta) + \frac{1}{\sqrt{\mu(\lambda - \epsilon^-)m}}\]
since $\mu$ is monotonically increasing.

\noindent{\bf Case 4: $\lambda + \epsilon^+\leq \lambda' \leq 1$}
By the probabilistic domination of Equation \eqref{eq: prob dom}, we may reduce to the situation of Case 3.

\end{proof}

Our main non-asymptotic theorem follows directly.

\begin{proof}[Proof of Theorem \ref{thm: main non-asymptotic}]
Note that
\[ (\text{Number of type I errors}) +  (\text{Number of type II errors}) = \sum_{x \in \X} \one_{G_x^c}.\]
The proof then follows by taking the expectation and then the supremum and then dividing by $n$.
\end{proof}

\subsection{Proof of asymptotic main theorem: Theorem \ref{thm: main}}

Once again, this will come from manipulating the result of Lemma \ref{lem: non-asymptotic return x}.  We have the following asymptotic characterization of the parameters of this Lemma.

\begin{lemma}\label{lem: asymptotic parameters}
Let $\lambda \in (2r -1, 1)$ also satisfy $ \lambda > 0$ and set $\epsilon$ as in Theorem \ref{thm: main}.  Let
\begin{equation}\label{eq: defines alphas}
\eta^- := \frac{\mu(\lambda) - \mu(\lambda - \epsilon)}{\sigma(\lambda - \epsilon)}\cdot\sqrt{m}, \qquad \eta^+ := -\frac{\mu(\lambda) - \mu(\lambda + \epsilon)}{\sigma(\lambda + \epsilon)}\cdot\sqrt{m}
\end{equation}

Then
\[\eta^+ \sim \eta \sim \eta^-.\]
Furthermore, 
\[\frac{1}{\sqrt{\mu(\lambda - \epsilon)m}} \rightarrow 0.\] 
\end{lemma}

\begin{remark}[A note regarding parameters in Theorem \ref{thm: main non-asymptotic}]\label{rem: epsilons}
In passing, we note that the asymptotic equivalence $\eta^+ \sim \eta \sim \eta^-$, combined with the fact that $\epsilon$ is proportional to $\eta$, may be manipulated to show that $\epsilon \sim \epsilon^+ \sim \epsilon^-$, where $\epsilon^+, \epsilon^-$ are defined in Theorem \ref{thm: main non-asymptotic}.  Thus, since $\epsilon \rightarrow 0$, we also have $\epsilon^+, \epsilon^- \rightarrow 0$.
\end{remark}
\begin{proof}
We write $\epsilon_m = \epsilon$ to emphasize dependence on $m$.  We will show that $\eta^- \sim \eta$.  The argument that $\eta^+ \sim \eta$ is quite similar.  
We control the numerator of $\eta^-$ via the first order Taylor approximation in $\epsilon_m$
\begin{equation}\label{eq: Taylor}
\mu(\lambda) - \mu(\lambda - \epsilon_m) = \epsilon_m \frac{d}{d \lambda} \mu(\lambda) + O(\epsilon_m^2) \sup_{t \in [\lambda -\epsilon_m, \lambda]} \frac{d^2}{dt^2} \mu(t).
\end{equation}
Note, while $\epsilon_m \rightarrow 0$, it is not apriori obvious that the first term dominates asymptotically since $\mu(t)$ also depends on $m$.  However, this will become apparent with a bit of calculus. 

We begin with the observation that 
\[\frac{d}{dt} \mu(t) = \frac{1}{2 \pi \sqrt{1 - t^2}} \exp\left(\frac{-h^2}{1 + t}\right).\]
which may be found in \cite{plackett1954reduction}.  A bit of calculus then gives the second derivative
\[\frac{d^2}{dt^2} \mu(t) = \left(\frac{h^2}{(1 + t)^2} + \frac{t}{1 - t^2} \right)\frac{d}{dt} \mu(t).\]
Note that this is positive for $t > 0$, and thus the first derivative is increasing.  Further, for $m$ large, $\lambda- \epsilon_m > 0$ since $\lambda > 0$ by assumption and $\epsilon_m \rightarrow 0$. We use these observations to develop the above equation into a bound on the supremum of the second derivative
\[\sup_{t \in [\lambda -\epsilon_m, \lambda]} \frac{d^2}{dt^2} \mu(t) \leq \left(h^2 + \frac{\lambda}{1 - \lambda^2}\right)\frac{d}{d\lambda} \mu(\lambda).\]

The first-order Taylor approximation \eqref{eq: Taylor} then becomes
\[\mu(\lambda) - \mu(\lambda - \epsilon_m) = \epsilon_m \frac{d}{d \lambda} \mu(\lambda) + O(\epsilon_m h^2) \epsilon_m \frac{d}{d \lambda} \mu(\lambda).\]
Thus, since $\epsilon_m h^2 \rightarrow 0$, we have
\begin{equation}
\label{eq: mu difference}
\mu(\lambda) - \mu(\lambda - \epsilon_m) \sim \epsilon_m \frac{d}{d\lambda} \mu(\lambda) = \epsilon_m \frac{1}{2 \pi \sqrt{1 - \lambda^2}} \exp\left(\frac{-h^2}{1 + \lambda}\right).
\end{equation}
Thus the numerator of $\eta^-$ is asymptotically equivalent to the right-hand side of the above expression, multiplied by $\sqrt{m}$.  Let us also note that Equation \eqref{eq: mu difference} combined with Equation \eqref{eq: expected nonzeros} imply that
\begin{equation}\label{eq: mu and little mu}
\mu(\lambda - \epsilon_m) \sim \mu(\lambda).
\end{equation}

We now move to the denominator of $\eta^-$, that is, $\sigma(\lambda - \epsilon_m) = \sqrt{\mu(\lambda - \epsilon_m) (1 - \mu(\lambda - \epsilon_m)}$.  It is not hard to show that $\mu(\lambda - \epsilon_m) \rightarrow 0$, and thus $\sigma(\lambda - \epsilon_m) \sim \sqrt{\mu(\lambda - \epsilon_m)}$. 

Thus, by Equation \eqref{eq: mu and little mu}, we have
\begin{equation}
\label{eq: mu(lambda - epsilon)}
\sigma(\lambda - \epsilon) \sim \sqrt{\mu(\lambda)} \sim \sqrt{\frac{(1 + \lambda)^2}{2 \pi h^2 \sqrt{1 - \lambda^2}} \exp\left(-\frac{h^2}{1 + \lambda}\right)}.
\end{equation}
Now we have shown that the right-hand side of Equation \eqref{eq: mu difference}, multiplied by $\sqrt{m}$, is asymptotically equivalent to the numerator of $\eta^-$ and the right-hand side of Equation \eqref{eq: mu(lambda - epsilon)} is asymptotically equivalent to the denominator.  If you divide the former by the latter, you get $\eta$, thus showing that $\eta^- \sim \eta$ as desired.

We complete the proof of the lemma by showing that $\frac{1}{\sqrt{\mu(\lambda - \epsilon) m}} \rightarrow 0$.  First, Equation \eqref{eq: mu and little mu} imply that this quantity is asymptotically equivalent to $\frac{1}{\sqrt{\mu(\lambda) m}}$.  As we have seen before the latter quantity converges to $0$ based on Equation \eqref{eq: expected nonzeros}.
\end{proof}

The above lemma implies the following result when $|\X| = 1$.

\begin{lemma}\label{lem: G_x}
Fix $y \in \sphere^{d-1}, \lambda \in (2r -1, 1)$ and $\eta > 0$.  Let $\epsilon$ satisfy Equation \eqref{eq: epsilon}.
Consider the (good) event $G_x := $\{if $\ip{x}{y} \geq \lambda + \epsilon$, then $S(x,y) \geq \mu(\lambda)$; if $\ip{x}{y} \leq \lambda - \epsilon$, then $S(x,y) \leq \mu(\lambda)$.\}
Then
\[\lim_{m \rightarrow \infty} \sup_{x \in \sphere^{d-1}} P(G_x^c) = P(N(0,1) \geq \eta).\]
\end{lemma}
\begin{proof}
By following the same steps as in the proof of Lemma \ref{lem: non-asymptotic G} we have
\begin{align*}
&\sup_{x \in \sphere^{d-1}} P(G_x^c) \geq P(N(0,1) \geq \eta^-) - \frac{1}{\sqrt{\mu(\lambda - \epsilon) m}},\\
&\sup_{x \in \sphere^{d-1}} P(G_x^c) \leq P(N(0,1) \geq \min(\eta^-, \eta^+)) + \frac{1}{\sqrt{\mu(\lambda - \epsilon)m}}.
\end{align*}
The proof of the lemma now follows from the continuity of $P(N(0,1) \geq t)$ as a function of $t$, combined with the asymptotic characterization of parameters in Lemma \ref{lem: asymptotic parameters}.
\end{proof}

We are now in position to prove our main non-asymptotic theorem.

\begin{proof}[Proof of Theorem \ref{thm: main non-asymptotic}]
This is precisely the same as the proof of Theorem \ref{thm: main}.  Note that
\[ (\text{Number of type I errors}) +  (\text{Number of type II errors}) = \sum_{x \in \X} \one_{G_x^c}.\]
The proof then follows by taking the expectation and then the supremum and then dividing by $n$.
\end{proof}

\end{document}

%% file: macros.tex
\usepackage{fullpage}
\usepackage{amsmath, amssymb, enumerate, color, graphicx, bm, url, mathbbol, amsthm}

\newtheorem{theorem}{Theorem}[section]

\newtheorem{lemma}[theorem]{Lemma}

\newtheorem{definition}[theorem]{Definition}

\newtheorem{remark}[theorem]{Remark}

\numberwithin{equation}{section}

\newcommand{\ip}[2]{\left\langle#1,#2\right\rangle}

\DeclareMathOperator*{\E}{\mathbb{E}}

\def \R {\mathbb{R}}

\def \one {{\bf 1}}

\def \moo| {\langle}
\def \< {\langle }
\def \> {\rangle }
\def \^ {\widehat}

\newcommand{\norm}[1]{\left \|#1\right \|}

\newcommand{\mat}[1]{\bm{#1}}

\def \E {\mat{E}}

\def \R {\mat{R}}

\def \X {\mat{X}}

%% file: macros.local.tex
\usepackage{hyperref}
\usepackage{graphicx}
\usepackage{caption}
\usepackage{subcaption}

\def \X {\mathcal{X}}

\def \bigO {\mathcal{O}}

\def \real {\mathbb{R}}

\def \integer {\mathbb{Z}}

\def \sphere {\mathbb{S}}

\newcommand{\score}[2]{S\left(#1, #2\right)}